\documentclass[11pt]{article}

\usepackage{a4}
\usepackage{authblk}
\usepackage{amsmath,amssymb}
\usepackage{amsthm}
\usepackage[dvipdfmx]{graphicx}

\pdfoutput=1

\usepackage{thmtools}

\newtheorem{theorem}{Theorem}
\newtheorem{lemma}{Lemma}

\theoremstyle{definition}

\newtheorem{remark}{Remark}
\newtheorem{example}{Example}

\usepackage[final,mode=multiuser]{fixme}
\fxuseenvlayout{colorsig}
\fxusetargetlayout{color}
\FXRegisterAuthor{hf}{ahf}{HF}
\FXRegisterAuthor{yn}{ayn}{YN}
\FXRegisterAuthor{si}{asi}{SI}
\fxsetface{env}{}

\newtheorem{fact}{Fact}

\newcommand{\Prefix}{\mathsf{Prefix}}
\newcommand{\Substr}{\mathsf{Substr}}
\newcommand{\Suffix}{\mathsf{Suffix}}

\newcommand{\Long}{\mathsf{long}}
\newcommand{\CDAWG}{\mathsf{CDAWG}}

\newcommand{\EndPos}{\mathsf{EndPos}}
\newcommand{\BegPos}{\mathsf{BegPos}}

\newcommand{\EqrL}{\equiv^{\mathrm{L}}}
\newcommand{\EqrR}{\equiv^{\mathrm{R}}}
\newcommand{\EqcL}[1]{[{#1}]^{\mathrm{L}}}
\newcommand{\EqcR}[1]{[{#1}]^{\mathrm{R}}}

\newcommand{\LeftM}{\mathsf{LeftM}}
\newcommand{\RightM}{\mathsf{RightM}}
\newcommand{\M}{\mathsf{M}}
\newcommand{\D}{\mathsf{d}}

\newcommand{\FIns}{\mathsf{f}_{\mathrm{Ins}}}
\newcommand{\GIns}{\mathsf{g}_{\mathrm{Ins}}}
\newcommand{\FDel}{\mathsf{f}_{\mathrm{Del}}}
\newcommand{\FSub}{\mathsf{f}_{\mathrm{Sub}}}
\newcommand{\GSub}{\mathsf{g}_{\mathrm{Sub}}}

\newcommand{\size}{\mathsf{e}}

\newcommand{\ASIns}{\mathsf{AS}_{\mathrm{LeftIns}}}
\newcommand{\ASDel}{\mathsf{AS}_{\mathrm{LeftDel}}}
\newcommand{\ASSub}{\mathsf{AS}_{\mathrm{LeftSub}}}

\begin{document}

\title{Tight bounds for the sensitivity of CDAWGs with left-end edits}

\author[1]{Hiroto~Fujimaru}
\author[2]{Yuto~Nakashima}
\author[2]{Shunsuke~Inenaga}

\affil[1]{Department of Information Science and Technology
{\tt fujimaru.hiroto.134@s.kyushu-u.ac.jp}}

\affil[2]{Department of Informatics, Kyushu University, Japan
{\tt \{nakashima.yuto.003, inenaga.shunsuke.380\}@m.kyushu-u.ac.jp}}

\date{}
\maketitle

\begin{abstract}
\emph{Compact directed acyclic word graphs} (\emph{CDAWGs}) [Blumer et al. 1987] are a fundamental data structure on strings with applications in text pattern searching, data compression, and pattern discovery. Intuitively, the CDAWG of a string $T$ is obtained by merging isomorphic subtrees of the suffix tree [Weiner 1973] of the same string $T$, thus CDAWGs are a compact indexing structure. In this paper, we investigate the sensitivity of CDAWGs when a single character edit operation (insertion, deletion, or substitution) is performed at the left-end of the input string $T$, namely, we are interested in the worst-case increase in the size of the CDAWG after a left-end edit operation. We prove that if $\size$ is the number of edges of the CDAWG for string $T$, then the number of new edges added to the CDAWG after a left-end edit operation on $T$ does not exceed $\size$. Further, we present a matching lower bound on the sensitivity of CDAWGs for left-end insertions, and almost matching lower bounds for left-end deletions and substitutions. We then generalize our lower-bound instance for left-end insertions to \emph{leftward online construction} of the CDAWG, and show that it requires $\Omega(n^2)$ time for some string of length $n$.
\end{abstract}

\section{Introduction}

\emph{Compact directed acyclic word graphs} (\emph{CDAWGs})~\cite{Blumer1987}  are a fundamental data structure on strings that have applications in fields including text pattern searching~\cite{Crochemore1997,Inenaga2005}, data compression~\cite{BelazzouguiC17,Takagi2017}, and pattern discovery~\cite{Takeda2000}.
Intuitively, the CDAWG of a string $T$, denoted $\CDAWG(T)$, is obtained by merging isomorphic subtrees of the suffix tree~\cite{Weiner1973} of the same string $T$. Thus the size of the CDAWG is never larger than that of the suffix tree. A more detailed analysis is reviewed below:

\sinote*{modified}{%
It is well known that the internal nodes of $\CDAWG(T)$ correspond to \emph{maximal repeats} in $T$, and the number $\size$ of right-extensions of maximal repeats in $T$ is equal to the number of edges of $\CDAWG(T)$.
This contrasts that the internal nodes of the suffix tree for $T$ corresponds to \emph{right-maximal repeats} in $T$.
While the suffix tree for any string $T$ of length $n$ contains
$\Theta(n)$ nodes and edges with unique end-marker $\$$ at the right-end of $T$,
it is known that the numbers of nodes and edges in the CDAWGs can be as small as $\Theta(\log n)$ for highly repetitive strings~\cite{RadoszewskiR12} even with $\$$.
}%

The number $\size$ of edges in $\CDAWG(T)$ has been used as one of repetitiveness measures of string $T$. Namely, when $\size$ is small, then the string contains a lot of repetitive substrings hence being well compressible.
Further, one can obtain a \emph{grammar-based compression} of size $O(\size)$ via the CDAWG of the input string $T$~\cite{BelazzouguiC17}. Some relations between $\size$ and the number $\mathsf{r}$ of equal-letter runs in the \emph{Burrows-Wheeler transform} (\emph{BWT})~\cite{BurrowsWheeler} have also been investigated~\cite{BelazzouguiCGPR15}.

Recently, Akagi et al.~\cite{AkagiFI2023} proposed the notion of \emph{sensitivity} of string repetitiveness measures and string compressors, including the aforementioned $\size$ and $\mathsf{r}$, the smallest \emph{string attractor} size $\gamma$~\cite{KempaP18}, the \emph{substring complexity} $\delta$~\cite{Kociumaka2023}, and the Lempel-Ziv parse size $\mathsf{z}$~\cite{LZ77}. The sensitivity of a repetitiveness measure $\mathsf{c}$ asks how much the measure size increases when a single-character edit operation is performed on the input string, and thus the sensitivity allows one to evaluate the robustness of the measure/compressor against errors/edits.

This paper investigates the sensitivity of CDAWGs when a single character edit operation (insertion, deletion, or substitution) is performed at the left-end of the input string $T$, namely, we are interested in the worst-case increase in the size of the CDAWG after an left-end edit operation. We prove that if $\size$ is the number of edges of the CDAWG for string $T$, then the number of new edges which are added to the CDAWG after an left-edit operation on $T$ is always less than $\size$. Further, we present a matching lower bound on the sensitivity of CDAWGs for left-end insertions, and almost matching lower bounds for left-end deletion, and substitution (see Table~\ref{tb:summary} for a summary of these results).

We then generalize our lower-bound instances for left-end insertion to \emph{leftward online construction} of the CDAWG, and show that it requires $\Omega(n^2)$ time.
Here, leftward online construction of the CDAWG for the input string $T$ of length $n$ refers to the task of updating the CDAWG of $T[i+1..n]$ to the CDAWG of $T[i..n]$ for decreasing $i = n, \ldots, 1$. 
This contrasts with the case of \emph{rightward online CDAWG construction} for which a linear-time algorithm exists~\cite{Inenaga2005}.

\begin{table}[h]
  \centering
  \caption{Our results: additive sensitivity of CDAWGs with left-end edit operations.}
  \label{tb:summary}
  \begin{tabular}{|c||c|c|}
    \hline
    edit operation & upper bound & lower bound \\ \hline
    left-end insertion~($T \Rightarrow aT$) & $\size-1$ & $\size-1$ \\ \hline
    left-end deletion~($T \Rightarrow T[2..|T|]$) & $\size-3$ & $\size-4$ \\ \hline
    left-end substitution~($T = aS \Rightarrow bS = T'$) & $\size$ & $\size-3$ \\ \hline
  \end{tabular}
\end{table}

A preliminary version of this work appeared in~\cite{FujimaruNI23_WORDS}.
Below is a list of new results in this full version:
\begin{itemize}
\item Full proofs for our lower bounds for the sensitivity of CDAWGs with left-end edit operations (Section~\ref{sec:insert_lowerbound} for insertions, Section~\ref{sec:delete_lowerbound} for deletions, and Section~\ref{sec:substitute_lowerbound} for substitutions).

\item The tight lower bound $\size-1$ for left-end insertions, which improves the previous lower bound $\size-2$ reported in the preliminary version~\cite{FujimaruNI23_WORDS}.

\item A tighter upper bound $\size-3$ for the sensitivity of CDAWGs with left-end deletions (Section~\ref{sec:deletions}), which improves the previous upper bound $\size-2$ reported in the preliminary version~\cite{FujimaruNI23_WORDS}.
  
\item A new  $\Omega(n^2)$-time lower bound for leftward online \emph{batched} constructions for CDAWGs, where a string of fixed length $b$ is prepended to the current string, and the task is to update the CDAWG of $T[1+kb..n]$ to the CDAWG of $T[1+(k-1)b..n]$ for decreasing $k = n/b, \ldots, 1$ (Theorem~\ref{theo:batched_online} in Section~\ref{sec:batched_online}).
\end{itemize}

\subsection*{Related work}
Akagi et al.~\cite{AkagiFI2023} presented lower bounds
when a new character is deleted (resp. substituted) in the middle of the string, with a series of strings for which the size $\size$ of the CDAWG additively increases by $\size-4$ (resp. $\size-2$).
They also showed a lower bound when a new character is inserted at the \emph{right-end} of the string, showing a series of strings for which the size of the CDAWG additively increases by $\size-2$.
While an additive $\size+O(1)$ upper bound for the case of right-end insertion
readily follows from the \emph{rightward} online construction of CDAWGs~\cite{Inenaga2005}, no non-trivial upper bounds for the other edit operations,
including our case of left-end edit operations, are known.

Our $\Omega(n^2)$ lower-bound for leftward online construction of the CDAWG extends the quadratic lower-bound for maintaining the CDAWG in the sliding window model~\cite{SenftD08} (remark that fixing the right-end of the sliding window is equivalent to our leftward online construction).

\section{Preliminaries}

Let $\Sigma$ be an \emph{alphabet} of size $\sigma$.
An element of $\Sigma^*$ is called a \emph{string}.
For a string $T \in \Sigma^*$, the length of $T$ is denoted by $|T|$.
The \emph{empty string}, denoted by $\varepsilon$, is the string of length $0$.
Let $\Sigma^+ = \Sigma^* \setminus \{\varepsilon\}$.
If $T = uvw$, then $u$, $v$, and $w$ are called a \emph{prefix}, \emph{substring},
and \emph{suffix} of $T$, respectively.
The sets of prefixes, substrings, and suffixes of string $T$ are denoted by
$\Prefix(T)$, $\Substr(T)$, and $\Suffix(T)$, respectively.
For a string $T$ of length $n$, $T[i]$ denotes the $i$th character of $T$
for $1 \leq i \leq n$,
and $T[i..j] = T[i] \cdots T[j]$ denotes the substring of $T$ that begins at position $i$ and ends at position $j$ on $T$ for $1 \leq i \leq j \leq n$.
For two strings $u$ and $T$,
let $\BegPos(u, T) = \{i \mid T[i..i+|u|-1] = u\}$ and
$\EndPos(u, T) = \{i \mid T[i-|u|+1..i] = u\}$ denote
the sets of beginning positions and the set of ending positions of $u$ in $T$,
respectively.

For any substrings $u, v \in \Substr(T)$ of a string $T$,
we write $u \EqrL_T v$ iff $\EndPos(u, T) = \EndPos(v, T)$.
Let $\EqcL{\cdot}_T$ denote the equivalence class of strings under $\EqrL_T$.
For $x \in \Substr(T)$, let $\Long(\EqcL{x}_T)$ denote the longest member of $\EqcL{x}_T$.
Let $\LeftM(T) = \{\Long(\EqcL{x}_T) \mid x \in \Substr(T)\}$.
Any element $u \in \LeftM(T)$ is said to be \emph{left-maximal} in $T$, since there are two distinct characters $c, d \in \Sigma$ such that $cu,du \in \Substr(T)$, or $u \in \Prefix(T)$.
For any non-longest element $y \in \EqcL{x}_T \setminus \{\Long(\EqcL{x}_T)\}$
there exists a unique non-empty string $\alpha$ such that $\alpha y = \Long(\EqcL{x}_T)$, i.e. any occurrence of $y$ in $T$ is immediately preceded by $\alpha$.

Similarly, we write $u \EqrR_T v$ iff $\BegPos(u, T) = \BegPos(v, T)$.
Let $\EqcR{\cdot}_T$ denote the equivalence class of strings under $\EqrR_T$.
For $x \in \Substr(T)$, let $\Long(\EqcR{x}_T)$ denote the longest member of $\EqcR{x}_T$.
Let $\RightM(T) = \{\Long(\EqcR{x}_T) \mid x \in \Substr(T)\}$.
Any element $u \in \RightM(T)$ is said to be \emph{right-maximal} in $T$, since there are two distinct characters $c, d \in \Sigma$ such that $uc,ud \in \Substr(T)$, or $u \in \Suffix(T)$.
For any non-longest element $y \in \EqcR{x}_T \setminus \{\Long(\EqcR{x}_T)\}$
there exists a unique non-empty string $\beta$ such that $y \beta = \Long(\EqcR{x}_T)$, i.e. any occurrence of $y$ in $T$ is immediately followed by $\beta$.
Let $\mathsf{M}(T) = \LeftM(T) \cap \RightM(T)$.
Any element of $\mathsf{M}(T)$ is said to be \emph{maximal} in $T$.

The \emph{compact directed acyclic word graph} (\emph{CDAWG}) of a string $T$,
denoted $\CDAWG(T) = (\mathsf{V}, \mathsf{E})$, is an edge-labeled DAG such that
\begin{eqnarray*}
  \mathsf{V}_T & = & \{\EqcL{x}_T \mid x \in \RightM(T)\}, \\
  \mathsf{E}_T & = & \{(\EqcL{x}_T, \beta, \EqcL{x \beta}_T) \mid \beta \in \Sigma^+, x, x\beta \in \RightM(T), y \beta \in \EqcL{x \beta}_T \mbox{ for any } y \in \EqcL{x}_T\}.
\end{eqnarray*}

\begin{figure}[t]
  \centering
  \includegraphics[keepaspectratio,scale=0.45]{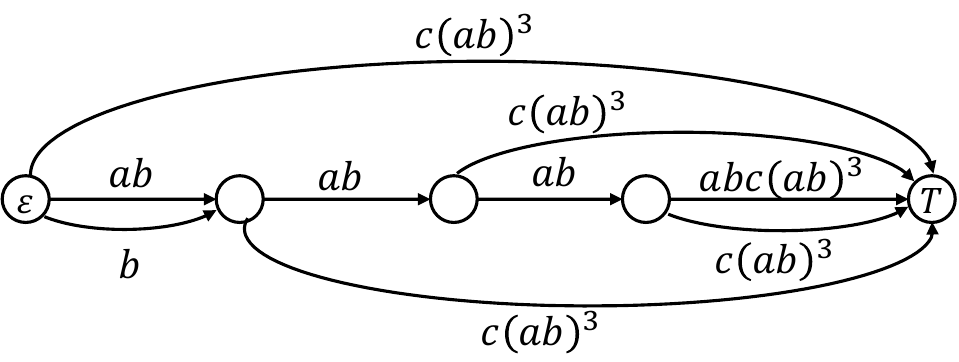}
  \caption{Illustration for $\CDAWG(T)$ of string $T=(ab)^4 c(ab)^3$. Every substring of $T$ can be spelled out from a distinct path from the source $\varepsilon$. There is a one-to-one correspondence between the maximal substrings in $\M(T) = \{\varepsilon, ab, (ab)^2, (ab)^3, (ab)^4 c(ab)^3\}$ and the nodes of $\CDAWG(T)$. The number of right-extensions of $\CDAWG(T)$ is the number $\size(T)$ of edges, which is 9 in this example.}
  \label{fig:CDAWG}
\end{figure}

See Figure~\ref{fig:CDAWG} for a concrete example of CDAWGs.
Intuitively, the strings in $\RightM(T)$ correspond to the
nodes of the suffix tree~\cite{Weiner1973} of $T$,
and the operator $\EqcL{\cdot}_T$ merges the isomorphic subtrees
of the suffix tree.
Recall that the nodes of the suffix tree for $T$
correspond to the right-maximal substrings of $T$.
Since $\Long(\EqcL{x}_T)$ is a maximal substring of $T$ for any $x \in \RightM(T)$, we have the following fact:
\hspace*{1pc}
\begin{fact} \label{fact:maximal_CDAWG}
There is a one-to-one correspondence between the elements of $\M(T)$ and the nodes of $\CDAWG(T)$.
\end{fact}
We can regard each element of $\M(T)$ as a node of $\CDAWG(T)$ by Fact~\ref{fact:maximal_CDAWG}.
We thus sometimes identify $\mathsf{V}_T$ with $\M(T)$ for convenience.
For any $x \in \M(T)$, $\D_{T}(x)$ denotes the out-degree of the node $x$ in $\CDAWG(T)$.

A non-empty substring $x$ of string $T$ is called a \emph{maximal repeat} in $T$
if $x$ is maximal in $T$ and $|\BegPos(x, T)| = |\EndPos(x, T)| \geq 2$.
We remark that the set of maximal repeats in $T$ coincides with
$\mathsf{M}(T) \setminus \{\varepsilon, T\}$,
namely the longest elements of all internal nodes of $\CDAWG(T)$ are maximal repeats in $T$, and they are the only maximal repeats in $T$.

The \emph{size} of $\CDAWG(T) = (\mathsf{V}_T, \mathsf{E}_T)$
for a string $T$ of length $n$
is the number $\size(T) = |\mathsf{E}_T|$ of edges in $\CDAWG(T)$,
which is also referred to as the number of right-extensions of maximal repeats in $T$.
Using this measure $\size$, we define
the worst-case additive \emph{sensitivity} of the CDAWG
with left-end edit operations (resp. insertion, deletion, and substitution) by:
\begin{eqnarray*}
  \ASIns(\mathsf{\size}, n) & = & \max_{T \in \Sigma^n, a \in \Sigma} \{\size(aT)-\size(T) \}, \\
  \ASDel(\size, n) & = & \max_{T \in \Sigma^n} \{\size(T[2..n])-\size(T)\}, \\
  \ASSub(\size, n) & = & \max_{T \in \Sigma^n, a \in \Sigma \setminus \{T[1]\}} \{\size(aT[2..n])-\size(T)\}.
\end{eqnarray*}

For the sensitivity of CDAWGs,
we first briefly describe the special case where both the original string $T$ and an edited string $T'$ are unary. Let $T = a^n$.
Clearly, every $a^i$ with $1 \leq i < n$ is a maximal substring of $T$
and it is only followed by $a$. Thus $\size(T) = n-1$.
In case of insertion, i.e. $T' = aT = a^{n+1}$,
we similarly have $\size(T') = n$.
Thus $\size(T') - \size(T) = 1$ for unary strings.
Symmetrically, we have $\size(T') - \size(T) = -1$ in the case of deletion
with $T' = a^{n-1}$.
There is no substitution when $\sigma = 1$.
In what follows, we focus on the case where $\sigma \geq 2$.

\section{Sensitivity of CDAWGs with left-end insertions}

We consider the worst-case additive sensitivity
$\ASIns(\size, n)$ of $\CDAWG(T)$ when a new character $a$ is prepended to
input string $T$ of length $n$, i.e. $T' = aT$.

In the following sections, we present tight bounds for $\ASIns(\size, n)$
in the case of left-end insertions.

\subsection{Upper bound for $\ASIns(\size, n)$ on CDAWGs}

We divide the value $\size(T') - \size(T)$ into two components
$\FIns(T)$ and $\GIns(T)$ such that
\begin{itemize}
\item $\FIns(T)$ is the total out-degrees of new nodes that appear in $\CDAWG(aT)$;
\item $\GIns(T)$ is the total number of new out-going edges of nodes that already exist in $\CDAWG(T)$.
\end{itemize}
Clearly $\size(T') - \size(T) \leq \FIns(T)+\GIns(T)$.
We first consider the above two components separately,
and then we merge them to obtain the desired upper bound.

\subsubsection{$\FIns(T)$: total out-degrees of new nodes}
Suppose $u$ is a new node for $\CDAWG(aT)$,
where $u \notin \M(T)$ and $u \in \M(aT)$.
This implies that there is a new occurrence of $u$ in $aT$ as a prefix.
Let $u = ax$. The following is our key lemma:

\vspace*{1pc}
\begin{lemma} \label{lem:node_created}
  If $ax \notin \M(T)$ and $ax \in \M(aT)$
  (i.e. $ax$ is a new node in $\CDAWG(aT)$), 
  then $x \in \M(T)$.
  Also, $\D_{aT}(ax) \leq \D_{T}(x)$.
\end{lemma}

\begin{proof}
Since $ax \in \Prefix(aT)$, $x \in \Prefix(T)$.
Thus $x$ is left-maximal in $T$.
Assume, for a contradiction, that $x$ is not right-maximal in $T$.
Then there exists a non-empty string $\beta \in \Sigma^+$ such that $x\beta = \Long(\EqcR{x}_T)$,
which means that any occurrence of $x$ in $T$ is immediately followed by $\beta$.
Thus $ax$ is also immediately followed by $\beta$ in $aT$, however,
this contradicts the precondition that $ax \in \M(aT)$.
Thus $x$ is right-maximal in $T$.
Since $\EndPos(ax, aT) \subseteq \EndPos(x, T)$, every right-extension of $ax$ in $aT$ is also right-extensions of $x$ in $T$.
Consequently, we have $\D_{aT}(ax) \leq \D_{T}(x)$.
\end{proof}

It follows from Lemma~\ref{lem:node_created} that the out-degree of each new node for $ax$
in $\CDAWG(aT)$ does not exceed the out-degree of the node for $x$ in $\CDAWG(T)$.
Also, there is an injective mapping from a new node $ax$ in $\CDAWG(aT)$
to an existing node $x$ in $\CDAWG(T)$ by Lemma~\ref{lem:node_created}.
Thus $\FIns(T) \leq \size(T)$ for any string $T$.

In the sequel, we give a tighter bound $\FIns(T) \leq \size(T)-1$ if $\size(T) \geq 3$.
For this purpose, we pick the case where $x = \varepsilon$,
assume that $ax = a$ becomes a new node in $\CDAWG(aT)$,
and compare the out-degree of the source $\varepsilon$ of $\CDAWG(T)$
and the out-degree of the new node $a$ in $\CDAWG(aT)$.
We consider the cases with $\sigma = 2$ and with $\sigma \geq 3$ separately:

\vspace*{1pc}
\begin{lemma} \label{lemma:binary_a_source}
Let $\sigma = 2$.
If
\begin{enumerate}
\item $a \notin \M(T)$, 
\item $a \in \M(aT)$, and
\item there exists a string $x \in \M(T) \setminus \{\varepsilon, T\}$
  such that $ax \notin \M(T)$ and $ax \in \M(aT)$,
\end{enumerate}
then $\D_{aT}(a) < \D_{T}(\varepsilon)$.
\end{lemma}

\begin{proof}
Let $\Sigma = \{a, b\}$.
We can exclude the case where $T = b^n$ due to the following reason:
Since $ab^i$ for each $1 \leq i < n$ is not maximal in $aT = ab^n$,
no new nodes are created in $\CDAWG(ab^n)$
(only a new edge labeled $ab^n$ from the source to the sink is created).

From now on, consider the case where $T$ contains both $a$ and $b$. This means that $\D_{T}(\varepsilon) = \sigma = 2$.
Since $a \in \M(aT)$, $a$ is a node of $\CDAWG(aT)$.
Assume, for a contradiction, that $\D_{aT}(a) = \D_{T}(\varepsilon)$.
We then have $\D_{aT}(a) = 2$, which means $aa, ab \in \Substr(aT)$.
There are two cases depending on the first character of $T$:
\begin{itemize}

\item If $T[1] = a$, then let $T = aw$.
Then, since $aT = aaw$, we have $ab \in \Substr(T)$.
Since $a \notin \M(T)$ (the first precondition),
$b$ is the only character that immediately follows $a$ in $T$,
meaning that $aa \notin \Substr(T)$.
Recall that the new node $ax$ must be a prefix of $aT = aaw$.
Since $x \neq \varepsilon$ (the third precondition), $|ax| \geq 2$,
and thus $aa$ is a prefix of $ax$.
However, since $aa \notin \Substr(T)$,
$aa$ occurs in $aT$ exactly once as a prefix and thus $ax$ occurs exactly once in $aT$.
This contradicts the third precondition that $ax$ is a new node in $\CDAWG(aT)$.

\item If $T[1] = b$, then we have that $ab \notin \Substr(T)$ by similar arguments as above. Thus $T$ must be of form $b^m a^{n-m}$ with $1 \leq m < n$.
Moreover, since $a \notin \M(T)$ and $a \in \M(aT)$ (the first and second preconditions), we have $T = b^{n-1}a$.
Then, for the edited string $aT = ab^{n-1}a$,
any new internal node $ax$ in $\CDAWG(aT)$ must be in form $ab^i$ with $1 \leq i < n$.
However, each $ax = ab^i$ occurs in $aT$ exactly once,
meaning that $\Long(\EqcR{ab^i}_{aT}) = aT$.
This contradicts the third precondition that $ax$ is a new node in $\CDAWG(aT)$.
\end{itemize}
Consequently, $\D_{aT}(a) < \D_{T}(\varepsilon)$.
\end{proof}

\vspace*{1pc}
\begin{lemma} \label{lemma:ternary_a_source}
Let $\sigma \geq 3$.
If $a \notin \M(T)$ and $a \in \M(aT)$,
then $\D_{aT}(a) < \D_{T}(\varepsilon)$.
\end{lemma}

\begin{proof}
By similar arguments to the proof for Lemma~\ref{lemma:binary_a_source},
we have that $T$ contains at least three distinct characters,
one of which is $a$. Thus $\D_{T}(\varepsilon) = \sigma \geq 3$.

Assume, for a contradiction, that $\D_{aT}(a) = \D_{T}(\varepsilon) = \sigma \geq 3$.
Since $a \notin \M(T)$ (i.e. $a$ is not maximal in $T$), we have the two following cases:
\begin{itemize}
\item If $a$ is not left-maximal in $T$, then $T[1] \neq a$ and there is a unique character $b$~($\neq a$) that immediately precedes $a$ in $T$, meaning that $aa \notin \Substr(T)$. Since $T[1] \neq a$, we also have $aa \notin \Substr(aT)$. Thus $\D_{aT}(a) < \sigma = \D_{T}(\varepsilon)$, a contradiction.

\item If $a$ is not right-maximal in $T$, then there is a unique character $b$ that immediately follows $a$ in $T$. The occurrence of $a$ as a prefix of $aT$ is followed by $T[1]$, and thus the number $\D_{aT}(a)$ of distinct characters following $a$ in $aT$ is at most $2 < \sigma = \D_{T}(\varepsilon)$, a contradiction.
\end{itemize}
Consequently, $\D_{aT}(a) < \D_{T}(\varepsilon)$. 
\end{proof}

By Lemmas~\ref{lemma:binary_a_source} and~\ref{lemma:ternary_a_source},
even if there appear new nodes $ax$ in $\CDAWG(aT)$
corresponding to all existing nodes $x$ in $\CDAWG(T)$,
we have a credit $\D_{T}(\varepsilon) - \D_{aT}(a) \geq 1$ in most cases.
The only exception is when $\sigma = 2$ and $\M(T) = \{\varepsilon, T\}$.
However, in this specific case $\CDAWG(T)$ consists only of the
two nodes (source and sink), namely $\size(T) = 2$.
Conversely, we have that the above arguments hold for any $\size(T) \geq 3$,
which leads to the following:

\vspace*{1pc}
\begin{lemma} \label{lemmma:new_nodes}
  For any string $T$ with $\size(T) \geq 3$, $\FIns(T) \leq \size(T)-1$.
\end{lemma}

\subsubsection{$\GIns(T)$: number of new branches from existing nodes}
The following lemma states that the out-degrees of most existing nodes
of $\CDAWG(T)$ do not change in $\CDAWG(aT)$,
except for a single unique node that can obtain a single new out-going edge in $\CDAWG(aT)$:

\vspace*{1pc}
\begin{lemma} \label{lemma:new_branches}
  For any $y \in \Substr(T)$ such that
  $y \in \M(T)$ and $y \in \M(aT)$, 
  $\D_{aT}(y) \in \{\D_{T}(y), \D_{T}(y)+1\}$.
  Also, there exists at most one substring $y$ with $\D_{aT}(y) = \D_{T}(y)+1$.
  Consequently $\GIns(T) \leq 1$.
\end{lemma}

\begin{proof}
  Since $y \in \M(T)$ and $y \in \M(aT)$, $y$ is a node in both $\CDAWG(T)$ and $\CDAWG(aT)$. Then we have that:
\[
\D_{aT}(y) =
\begin{cases}
  \D_{T}(y)+1 & \mbox{if } y \in \Prefix(aT) \mbox{ and $yb$ occurs in $aT$ only as a prefix}, \\
  \D_{T}(y) & \mbox{otherwise},
\end{cases}
\]
where $b$ is the character that immediately follows the occurrence of $y$ as a prefix of $aT$, namely $b = T[|y|]$. 

Assume, for a contradiction, that there exist two distinct substrings
$x, y \in \M(T) \cap \M(aT)$ such that $\D_{aT}(x) = \D_{T}(x)+1$
and $\D_{aT}(y) = \D_{T}(y)+1$.
Since both $x$ and $y$ must be distinct prefixes of $aT$,
we can assume w.l.o.g. that $|x| < |y|$,
which means that $x$ is a proper prefix of $y$.
Thus the occurrence of $x$ as a prefix of $aT$ is immediately followed by
the character $c = y[|x|+1]$.
We recall that $y$ occurs in $T$ since $y \in \M(T)$.
Therefore there is an occurrence of $x$ in $T$ that is immediately followed by $c$, which leads to $\D_{aT}(x) = \D_{T}(x)$, a contradiction.
\end{proof}

\subsubsection{Putting all together}
Due to Lemma~\ref{lemmma:new_nodes} and Lemma~\ref{lemma:new_branches},
we have an upper bound $\size(T') - \size(T) \leq \FIns(T)+\GIns(T) \leq \size(T)-1 + 1 = \size(T)$ for $\sigma \geq 2$.
We remark that the equality holds only if both of the following conditions are satisfied:
\begin{enumerate}
\item[(a)] For any $x \in \M(T) \setminus \{\varepsilon\}$, $ax \notin \M(T)$, $ax \in \M(aT)$, and $\D_{aT}(ax) = \D_{T}(x)$;
\item[(b)] There exists a unique string $x \in \Substr(T)$ such that $\D_{aT}(x) = \D_{T}(x)+1$.
\end{enumerate}
However, in the next lemma, we show that no strings $x$ can
satisfy both Conditions (a) and (b) simultaneously:

\vspace*{1pc}
\begin{lemma} \label{lemma:reduce_one_edge}
  If $ax \notin \M(T)$ and $ax \in \M(aT)$,
  then $\D_{aT}(x) = \D_{T}(x)$.
\end{lemma}

\begin{proof}
  Assume, for a contradiction, that $\D_{aT}(x) \neq \D_{T}(x)$.
  By Lemma~\ref{lemma:new_branches} we have that $\D_{aT}(x) = \D_{T}(x)+1$.
  Then, it also follows from the proof of Lemma~\ref{lemma:new_branches} that
  $x$ is a prefix of $aT$
  and the character $b = T[|x|]$ that immediately
  follows the prefix occurrence of $x$ in $aT$ differs from any other characters  that immediately follow the occurrences of $x$ in $T$.
  Namely, we have $b \notin \Sigma' = \{T[i+1] \mid i \in \EndPos(x, T)\}$.
  Moreover, by Lemma~\ref{lem:node_created}, $ax$ is also a prefix of $aT$.
  This means that $x$ is a prefix of $ax$, and hence $ax = xb$,
  which means that $x = a^{|x|}$ and $a = b$.
  Because $\sigma \geq 2$, $T \neq x$.
  Since $ax \in \M(aT)$ and $x \neq T$, $ax$~($= xb$) occurs in $T$.
  This means that $b = c$ for some $c \in \Sigma'$, a contradiction.
  Thus, $\D_{aT}(x) = \D_{T}(x)$.
\end{proof}

We have $\size(T) \geq 3$ only if $|T| \geq 3$. By wrapping up Lemma~\ref{lemmma:new_nodes}, Lemma~\ref{lemma:new_branches}, and Lemma~\ref{lemma:reduce_one_edge}, we obtain the main result of this subsection:

\vspace*{1pc}
\begin{theorem}
For any $n \geq 3$ and $\size \geq 3$, $\ASIns(\size,n) \leq \size-1$.
\end{theorem}

\subsection{Lower bound for $\ASIns(\size, n)$ on CDAWGs}
\label{sec:insert_lowerbound}

Below, we present a matching lower bound for
$\ASIns(\size, n)$ for the case of left-end insertions:

\vspace*{1pc}
\begin{theorem} \label{theo:lowerbound_insertion}
  There exists a family of strings $T$ such that
  $\size(T')-\size(T) = \size(T)-1$,
  where $T' = bT$ with $b \in \Sigma$.
  Therefore $\ASIns(\size, n) \geq \size-1$.
\end{theorem}

\begin{proof}
Consider string
$$T=(ab)^{m+1} c(ab)^m,$$
where $a, b, c \in \Sigma$. 
We have that
$$\M(T)=\{\varepsilon, ab, (ab)^2,...,(ab)^m, T\}.$$
Then, since $\D_T(\varepsilon) = 3$, $\D_T((ab)^i) = 2$ for every $1 \leq i \leq m$, and $\D_T(T) = 0$,
we have $\size(T) =2m+3$.

Let us now prepend character $b$ to $T$ and obtain
$$T' = b(ab)^{m+1} c(ab)^m.$$
We have that
\begin{eqnarray*}
  \M(T') & = & \{\varepsilon,ab,(ab)^2,...,(ab)^m,b,bab,b(ab)^2,...,b(ab)^m,T'\} \\
  & = & \left( \M(T) \setminus \{T\} \right) \cup \{b,bab,b(ab)^2,...,b(ab)^m\} \cup \{T'\},
\end{eqnarray*}
and that $\D_{T'}(\varepsilon)=3$,
$\D_{T'}((ab)^i) = 2$ for every $1 \leq i \leq m$,
$\D_{T'}(b(ab)^i) = 2$ for every $0 \leq i\leq m$, and $\D_{T'}(T') = 0$
(see Figure~\ref{fig:insert_lowerbound} for illustration).
Thus $\size(T') = 4m+5 = 2(2m+3)-1 = 2\size(T)-1$ which shows the claim
$\ASIns(\size,n) \geq \size - 1$.
\end{proof}

\begin{figure}[!h]
  \centering
  \includegraphics[keepaspectratio,scale=0.42]{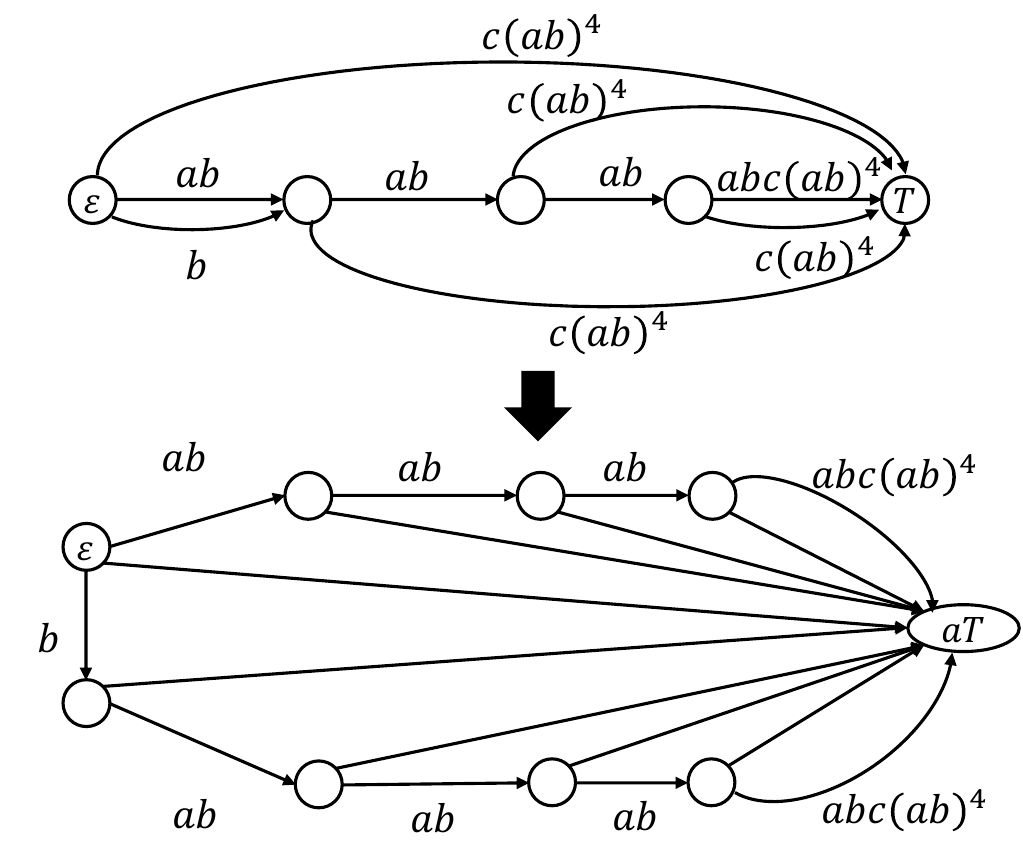}
  \caption{Illustration for the CDAWGs of strings $T=(ab)^3 abc(ab)^3$ and
    $T'= bT=b(ab)^3 abc (ab)^3$ with $m = 3$. The omitted edge labels are all $c(ab)^4$.
  Observe that $\size(T) = 9$ and $\size(T') = 17$, and hence $\size(T')-\size(T) = 8 = \size(T)-1$ with this left-end insertion.}
  \label{fig:insert_lowerbound}
\end{figure}

\section{Sensitivity of CDAWGs with left-end deletions}
\label{sec:deletions}

In this section we investigate the worst-case additive sensitivity
$\ASDel(\size, n)$ of $\CDAWG(T)$ when $T[1]$ is deleted from
the original input string $T$ of length $n$.

\subsection{Upper bound for $\ASDel(\size, n)$ on CDAWGs}

Let $a = T[1]$ be the first character of string $T$.
Let $T = aS$ and $T' = S$, and we consider left-end deletion $aS \Rightarrow S$.
Since deleting the left-end character from $T$ never increases
the right-contexts of any substring in $S$,
it suffices for us to consider $\FDel(T) = \FDel(aS)$,
the total out-degrees of new nodes that appear in $\CDAWG(T') = \CDAWG(S)$,
namely $\size(S) - \size(aS) \leq \FDel(aS)$.

%
Let $x$ be a new node in $\CDAWG(S)$. We have the following:

\vspace*{1pc}
\begin{lemma} \label{lem:node_created_del}
  If $x \notin \M(aS)$ and $x \in \M(S)$, 
  then $x \in \Prefix(S)$ and $ax \in \M(aS)$.
  Also, $\D_{S}(x) = \D_{aS}(ax)$.
\end{lemma}

\begin{proof}
  Since $x \notin \M(aS)$, $x$ is either not left-maximal or not right-maximal in $aS$.
  If $x$ is not right-maximal in $aS$, 
  then $x$ is also not right-maximal in $S$, hence $x \notin \M(S)$.
  However, this contradicts the precondition $x \in \M(S)$.
  Thus $x$ is not left-maximal in $aS$.
  Then, there exists a non-empty unique string $\alpha \in \Sigma^+$ such that $\alpha x = \Long(\EqcL{x}_{aS})$,
  which means that any occurrence of $x$ in $aS$ is immediately preceded by $\alpha$.
  Assume, for a contradiction, that $x \notin \Prefix(S)$.
  Since $x \in \M(S)$, $x = \Long(\EqcL{x}_{S}) = \Long(\EqcL{x}_{aS})$, however,
  this contradicts that $\alpha$ is a non-empty string.
  Thus $x \in \Prefix(S)$, and hence $ax \in \Prefix(aS)$.
  Since $ax \in \Prefix(aS)$ and $x$ is right-maximal in $aS$,
  $ax$ is a maximal string of $aS$. Thus $ax \in \M(aS)$.
  
  Then, we have that $x$ is not left-maximal in $aS$, which also means that $x \notin \Prefix(aS)$. Since $ax \in \Prefix(aS)$,
  $\EndPos(ax, aS) = \EndPos(x, aS) = \EndPos(x, S)$ holds.
  Consequently we have that $\D_{aS}(ax) = \D_{S}(x)$.
\end{proof}

By Lemma~\ref{lem:node_created_del}, the out-degree of each new node for $x$ in $\CDAWG(S)$ does not exceed the out-degree of the node for $ax$ in $\CDAWG(aS)$.
Also by Lemma~\ref{lem:node_created_del},
there is an injective mapping from a new node $x$ in $\CDAWG(S)$
to an existing node $ax$ in $\M(aS) \setminus \{\varepsilon\}$.
Since $\D_{aS}(\varepsilon) = \sigma \geq 2$,
the sum of out-degrees of all new nodes in $\CDAWG(S)$ is at most $\size(T) - \sigma = \size(T) - 2$.

Furthermore, we have another lemma for the existing nodes.

Let $x \in \Prefix(T)$ be the longest repeating prefix of $T$
such that $x$ occurs at least twice in $T$.
Since $x \in \Prefix(T)$, $x$ is left-maximal in $T$.
It also follows that $x$ is right-maximal in $T$,
since otherwise there is a non-empty string $\beta \in \Sigma^+$
such that $x\beta = \Long(\EqcR{x}_{T})$.
But this contradicts that $x$ is the longest repeating prefix of $T$.
Thus we have:

\vspace*{1pc}
\begin{lemma} \label{cor:LRP_maximal_repeat}
The longest repeating prefix of a string $T$ is a maximal repeat of $T$.
\end{lemma}

\vspace*{1pc}
Let $z$ be the longest repeating prefix of string $aS$,
where $a \in \Sigma$ and $S \in \Sigma^*$.
By Lemma~\ref{cor:LRP_maximal_repeat}, we have $z \in \M(aS)$.
We have the following lemma:

\vspace*{1pc}
\begin{lemma} \label{lem:egde_deleted_del}
 If $z$ is the longest repeating prefix of $aS$ and $z \in \M(S)$, 
 then $\D_{S}(z) < \D_{aS}(z)$.
\end{lemma}

\begin{proof}
  Let $b$ be the character that immediately follows the prefix $z$ in $aS$,
  namely $zb \in \Prefix(aS)$.
  Since $z$ is the longest repeating prefix of $aS$,
  $zb$ appears in $aS$ only as a prefix.
  Since $z$ is a maximal repeat of $aS$
  (by Lemma~\ref{cor:LRP_maximal_repeat}),
  $z$ is the longest string represented by the node $\Long(\EqcL{z}_{aS})$
  of $\CDAWG(aS)$.
  Thus $\CDAWG(aS)$ has an edge $e = (\EqcL{z}_{aS}, (aS)[|z|..n], \EqcL{aS}_{aS})$,
  where $S[|z|] = b$.
  Since $zb$ does not appear in $S$,
  the edge $e$ is removed when we delete the first character $a$ from $aS$.
  Thus, $\D_{S}(z) < \D_{aS}(z)$.
  We note that the above argument is valid also when $z = \varepsilon$,
  in which case $b = a$ and the out-edge beginning with $b$ is removed
  from the source of the CDAWG.
\end{proof}

By Lemma~\ref{lem:egde_deleted_del}, at least one edge must be deleted from $\CDAWG(T)$ after a left-end deletion on any string $T = aS$. 
Recalling that the sum of out-degrees of all new nodes does not exceed $\size(T) - 2$, we have:

\vspace*{1pc}
\begin{theorem}
  For any $n$, $\ASDel(\size,n) \leq \size-3$.
\end{theorem}
  
\subsection{Lower bound for $\ASDel(\size, n)$ on CDAWGs}
\label{sec:delete_lowerbound}

The next lower bound for $\ASDel(\size, n)$ holds.

\vspace*{1pc}
\begin{theorem} \label{theo:lowerbound_deletion}
  There exists a family of strings $T$ such that
  $\size(S)-\size(T) = \size(T)-4$,
  where $T = aS$ with $a \in \Sigma$.
  Therefore $\ASDel(\size, n) \geq \size-4$.
\end{theorem}

\begin{proof}
Consider string
$$T=(ab)^{m+1} c(ab)^m,$$
where $a, b, c \in \Sigma$. 
We have that
$$\M(T)=\{\varepsilon, ab, (ab)^2,...,(ab)^m, T\}.$$
Then, since $\D_T(\varepsilon) = 3$, $\D_T((ab)^i) = 2$ for every $1 \leq i \leq m$, and $\D_T(T) = 0$,
we have $\size(T) =2m+3$.

Let us delete the first character $a =T[1]$ from $T$ and obtain
$$T' = b(ab)^m c(ab)^m.$$
We have that
\begin{eqnarray*}
  \M(T') & = & \{\varepsilon,ab,(ab)^2,...,(ab)^m,b,bab,b(ab)^2,...,b(ab)^{m-1},T'\} \\
  & = & \left( \M(T) \setminus \{T\} \right) \cup \{b,bab,b(ab)^2,...,b(ab)^{m-1}\} \cup \{T'\},
\end{eqnarray*}
and that $\D_{T'}(\varepsilon)=3$,
$\D_{T'}((ab)^i) = 2$ for every $1 \leq i \leq m-1$,
$\D_{T'}(b(ab)^i) = 2$ for every $0 \leq i\leq m-1$,
$\D_{T'}((ab)^m) = 1$,
and $\D_{T'}(T') = 0$
(see Figure~\ref{fig:delete_lowerbound} for illustration).
Thus $\size(T') = 4m+2 = 2(2m+3)-4 = 2\size(T)-4$ which shows the claim
that $\ASDel(\size, n) \geq \size-4$.
\end{proof}

\begin{figure}[h!]
  \centering
  \includegraphics[keepaspectratio,scale=0.42]{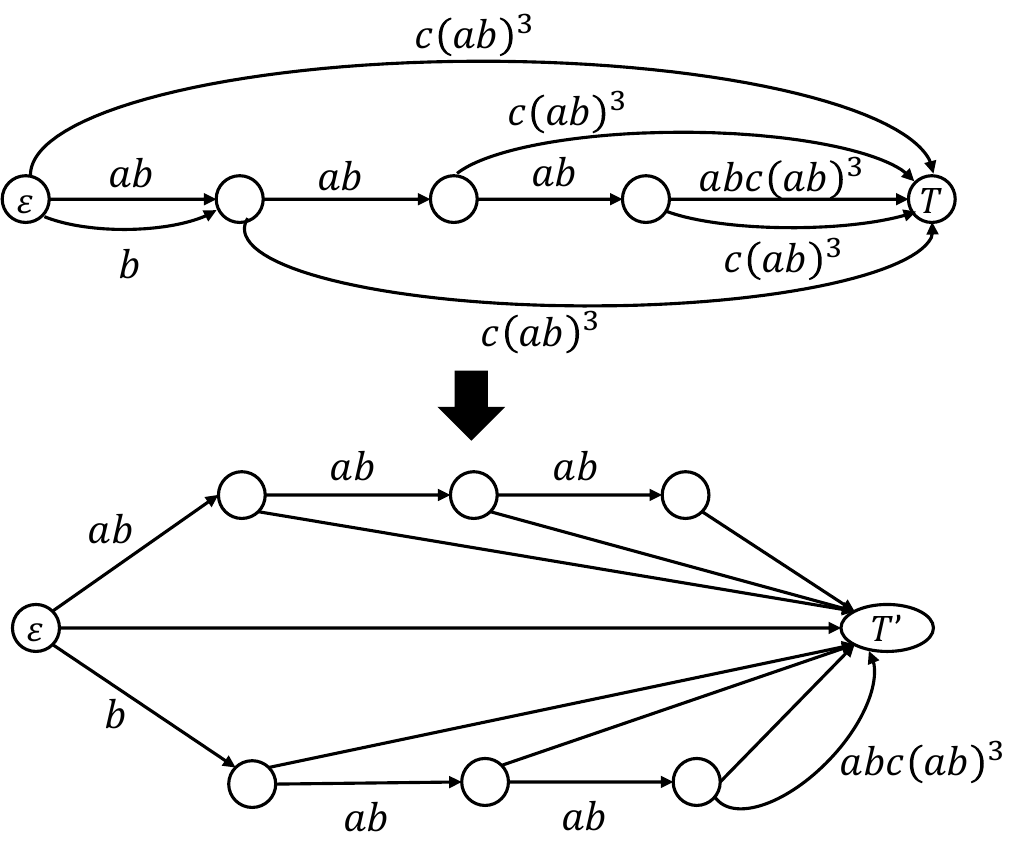}
  \caption{Illustration for the CDAWGs of strings $T=(ab)^4 abc(ab)^3$ and
   $T'= T[2..n]=b(ab)^3 c (ab)^3$ with $m = 3$. The omitted edge labels are all $c(ab)^3$. Observe that $\size(T) = 9$, $\size(T') = 14$, and hence $\size(T')-\size(T) = 5 = \size(T)-4$ with this left-end deletion.}
  \label{fig:delete_lowerbound}
\end{figure}

\section{Sensitivity of CDAWGs with left-end substitutions}

We consider the worst-case additive sensitivity
$\ASSub(\size, n)$ of $\CDAWG(T)$ when $T[1]$ is substituted by a new character $b \neq T[1]$, i.e. $T' = bT[2..n]$.

\subsection{Upper bound for $\ASSub(\size, n)$ on CDAWGs}

Similarly to the case of insertions,
we separate $\size(T') - \size(T)$ into the two following components
$\FSub(T)$ and $\GSub(T)$ such that
\begin{itemize}
\item $\FSub(T)$ is the total out-degrees of new nodes that appear in $\CDAWG(T')$;
\item $\GSub(T)$ is the total number of new out-going edges of nodes that already exist in $\CDAWG(T)$.
\end{itemize}
We regard a substitution as a sequence of a deletion and an insertion,
i.e. two consecutive edit operations such that $aS~(= T) \Rightarrow S \Rightarrow bS~(= bT[2..n] = T')$.

\subsubsection{$\FSub(T)$: total out-degrees of new nodes}

Let $u$ be a new node in $\CDAWG(bS)$ that does not exist in $\CDAWG(aS)$,
namely $u \in \M(bS)$ and $u \notin \M(aS)$.
We categorize each new node $u$ to the two following types $u_1$ and $u_2$ as:
\begin{enumerate}
  \item $u_1 \in \M(S)$ so that $u_1$ is generated by deletion $aS \Rightarrow S$;
  \item $u_2 \notin \M(S)$ so that $u_2$ is generated by insertion $S \Rightarrow bS$.
\end{enumerate}
Node $u_1$ is a new node that appears in $\CDAWG(S)$.
Thus, it follows from Lemma~\ref{lem:node_created_del}
that node $au_1$ exists in $\CDAWG(aS)$.
Since $u_2$ is not a node in $\CDAWG(S)$,
it follows from Lemma~\ref{lem:node_created} that
$u_2 = bx$ and $x$ is a node in $\CDAWG(S)$.
Based on this observation, we will show that there is an injective mapping from the new nodes in $\CDAWG(bS) = \CDAWG(T')$ to the existing nodes in $\CDAWG(aS) = \CDAWG(T)$.
In doing so, we must show that the two non-injective situations do not occur:
\begin{enumerate}
  \item[(i)] a new node $bx$ is generated by insertion $S \Rightarrow bS$, where $x$ is generated by deletion $aS \Rightarrow S$ and $x$ remains as a node in $\CDAWG(bS)$;
  \item[(ii)] a new node $bax$ generated by insertion $S \Rightarrow bS$, where $x$ is generated by deletion $aS \Rightarrow S$ and $x$ remains as a node in $\CDAWG(bS)$.
\end{enumerate}
Suppose (for a contradiction) that Case (i) happens.
Then, a new node $x$ is generated from an existing node $ax$, and $bx$ is generated from $x$.
Therefore, two new nodes could be generated from existed node $ax\in \M(aS)$.
However, the next lemma shows that this situation (Case (i)) does not occur unless $x = S$:

\vspace*{1pc}
\begin{lemma} \label{lem:injective_ax}
If $x \neq S$, $x\notin \M(aS)$, $x\in \M(S)$, and $x\in \M(bS)$, then $bx \notin \M(bS)$.
\end{lemma}

\begin{proof}
  Since $x \notin M(aS)$ and $x \in M(S)$, $x\in \Prefix(S)$ by Lemma~\ref{lem:node_created_del}.
  Since $x\in \M(S)$ and $ax \in \Prefix(aS)$,
  $ax \EqrL_{aS} x$ and $ax = \Long(\EqcL{x}_{aS})$.
  This means that $bx$ occurs exactly once in $bS$ as a proper prefix.
  Thus, $bx \notin \RightM(bS)$ which leads to $bx\notin \M(bS)$.
\end{proof}
As for Lemma~\ref{lem:injective_ax}, the situation (Case (i)) can occur if $x = S$. 
However, if $x = S$, then $S\in \M(bS)$ which implies that $S$ occurs in $bS$ as prefix $bS[1..(n-1)]$.
Thus, $S = b^n$, $T = aS = ab^n$ and $T' = bS = b^{n+1}$.
It is clear that $\size(aS) = \size(bS) = n+1$.
Therefore the size of the CDAWG does not change when $x = S$.

Now we turn our attention to Case (ii) and assume (for a contradiction) that it happens.
Then, two new nodes $bax$ and $x$ could be generated from a single existing node $ax$.
According to the following lemma, however, this situation cannot occur:

\vspace*{1pc}
\begin{lemma} \label{lem:injective_bax}
  If $ax\in \M(aS)$, $x \notin \M(aS)$, $bax\notin \M(aS)$, $x\in \M(S)$, and $bax\notin \M(S)$, 
  then $bax\notin\M(bS)$.
\end{lemma}

\begin{proof}
  Assume for a contradiction that $bax\in\M(bS)$.
  Since $x \notin M(aS)$ and $x \in M(S)$, $x\in\Prefix(S)$ by Lemma~\ref{lem:node_created_del}. 
  Also, since $bax \notin M(S)$ and $bax \in M(bS)$, $ax\in \Prefix(S)$ by Lemma~\ref{lem:node_created}.
  This means that $x\in \Prefix(ax)$ and $x=a^{|x|}$.
  Since $ax=a^{|x|+1}$ is a maximal substring of $aS$, $x$ is also a maximal substring of $aS$.
  Thus $x\in \M(aS)$, however, this contradicts the precondition that $x\notin \M(aS)$.
  Thus $bax\notin\M(bS)$.
\end{proof}

As a result, there is an injective mapping from the new nodes $u_1$ (resp. $u_2=bx$) in $\CDAWG(bS)$
to the existing nodes $au_1$ (resp. $x$) in $\CDAWG(aS)$ by Lemmas~\ref{lem:node_created}, \ref{lem:node_created_del}, \ref{lem:injective_ax}, and~\ref{lem:injective_bax}.
It also follows from these lemmas that the out-degree of each new node
in $\CDAWG(bS)$ does not exceed the maximum out-degree of $\CDAWG(aS)$.
Finally, we consider the source $\varepsilon$. By Lemmas~\ref{lemma:binary_a_source}, \ref{lemma:ternary_a_source}, and~\ref{lem:node_created_del}, if $b\in \M(bS)$, $b \notin \M(aS)$, and $\size(aS)\geq 3$, then $\D_{bS}(b) \leq \D_{aS}(\varepsilon)$.
Thus we have:

\vspace*{1pc}
\begin{lemma} \label{lem:new_nodes_sub}
  For any string $T$ with $\size(T) \geq 3$, $\FSub(T) \leq \size(T)-1$.
\end{lemma}

\subsubsection{$\GSub(T)$: number of new branches from existing nodes}
Since left-end deletions do not create new branches from existing nodes (recall Section~\ref{sec:deletions}), it is immediate from Lemma~\ref{lemma:new_branches} that:

\vspace*{1pc}
\begin{lemma} \label{lem:new_branches_sub}
  For any string $T$, $\GSub(T) \leq 1$.
\end{lemma}

\subsubsection{Putting all together}
Our main result of this section follows from Lemmas~\ref{lem:new_nodes_sub} and~\ref{lem:new_branches_sub}:

\vspace*{1pc}
\begin{theorem}
For any $n \geq 4$ and $\size \geq 3$, $\ASSub(\size,n) \leq \size$.
\end{theorem}

\subsection{Lower bound for $\ASSub(\size, n)$ on CDAWGs}
\label{sec:substitute_lowerbound}

The next lower bound for $\ASSub(\size, n)$ holds.

\vspace*{1pc}
\begin{theorem} \label{theo:lowerbound_substitution}
  There exists a family of strings $T$ such that
  $\size(T')-\size(T) = \size(T)-3$,  where $T' = bT[2..n]$ with $b \in \Sigma \setminus \{T[1]\}$.
  Therefore $\ASSub(\size, n) \geq \size-3$.
\end{theorem}

\begin{figure}[h!]
  \centering
  \includegraphics[keepaspectratio,scale=0.42]{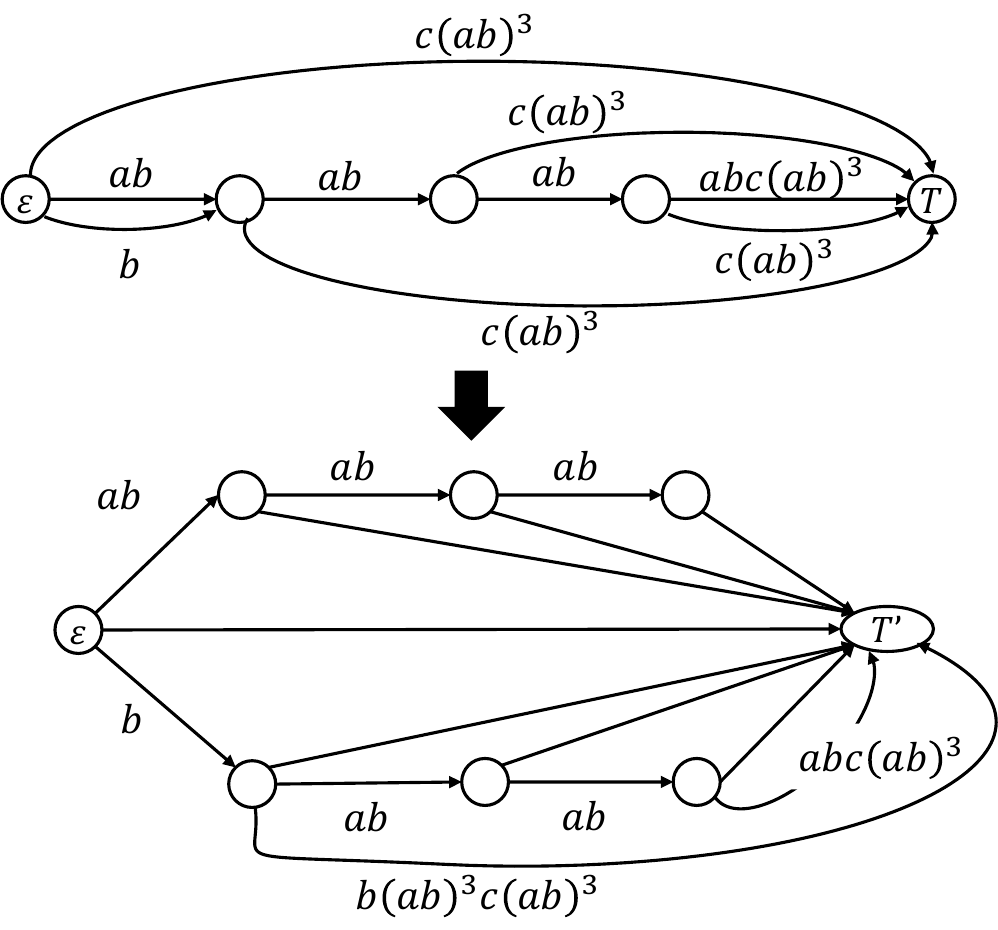}
  \caption{Illustration for the CDAWGs of strings $T=(ab)^4 c(ab)^3$ and
   $T'= bT[2..n]=bb(ab)^3 c (ab)^3$ with $m = 3$. The omitted edge labels are all $c(ab)^3$. Observe that $\size(T) = 9$, $\size(T') = 15$, and hence $\size(T')-\size(T) = 6 = \size(T)-3$ with this left-end deletion.}
  \label{fig:substitution_lowerbound}
\end{figure}

\begin{proof}
Consider string
$$T=(ab)^{m+1} c(ab)^m,$$
where $a, b, c \in \Sigma$. 
We have that
$$\M(T)=\{\varepsilon, ab, (ab)^2,...,(ab)^m, T\}.$$
Then, since $\D_T(\varepsilon) = 3$, $\D_T((ab)^i) = 2$ for every $1 \leq i \leq m$, and $\D_T(T) = 0$,
we have $\size(T) =2m+3$.

Let us now replace the first character $a = T[1]$ of $T$ by $b$ and obtain
$$T' = bb(ab)^m c(ab)^m.$$
We have that
\begin{eqnarray*}
  \M(T') & = & \{\varepsilon,ab,(ab)^2,...,(ab)^m,b,bab,b(ab)^2,...,b(ab)^{m-1},T'\} \\
  & = & \left( \M(T) \setminus \{T\} \right) \cup \{b,bab,b(ab)^2,...,b(ab)^{m-1}\} \cup \{T'\},
\end{eqnarray*}
and that $\D_{T'}(\varepsilon)=\D_{T'}(b)=3$,
$\D_{T'}((ab)^i) = \D_{T'}(b(ab)^i) = 2$ for every $1 \leq i \leq m-1$, 
$\D_{T'}(b(ab)^m) = 1$, 
and $\D_{T'}(T') = 0$
(see Figure~\ref{fig:substitution_lowerbound} for illustration).
Thus $\size(T') = 4m+3 = 2(2m+3)-3 = 2\size(T)-3$ which shows the claim
that $\ASSub(\size, n) \geq \size -3$.
\end{proof}

\section{Quadratic-time bounds for leftward online construction}

\subsection{Leftward online construction of CDAWGs}

The leftward online construction problem for the CDAWG is,
given a string $T$ of length $n$,
to maintain $\CDAWG(T[i..n])$ for decreasing $i = n, \ldots, 1$.
By extending our lower bound on the sensitivity with left-end insertions/deletions from Sections~\ref{sec:insert_lowerbound} and \ref{sec:delete_lowerbound}, a quadratic bound for this online CDAWG construction follows:

\vspace*{1pc}
\begin{theorem} \label{theo:lowerbound_leftonline}
  There exists a family of strings $T_m$ for which
  the total work for building $\CDAWG(T_m[i..n_m])$ for decreasing
  $i = n_m, \ldots, 1$ is $\Omega({n_m}^2)$, where $n_m = |T_m|$.
\end{theorem}

\begin{proof}
  Consider string
  $$T_m=(ab)^{2m} cab (ab)^{2m}\$, $$
  where $a,b,c,\$\in \Sigma$.
  For $0 \leq k \leq m$,
  let $T_{k,m}$ denote a series of suffixes of $T_{m}$
  such that
  $$T_{k,m} =(ab)^{m+k}cab(ab)^{2m} \$.$$
  Notice $T_{m,m} = T_m$,
  $m = \Theta(n_m)$ with $n_m = |T_{m,m}|$, and $T_{k,m}=T_m[2(m-k)+1..n_m]$.

  Now, we consider building $\CDAWG(T_m[i..n_m])$ for decreasing $i=n_m, \ldots, 1$,
  and suppose we have already built $\CDAWG(T_{k,m})$.  
  For this string $T_{k,m}$, we have that
  $\M(T_{k,m})=\{\varepsilon, ab, (ab)^2, \ldots ,(ab)^{2m},T_{k,m}\}$.
  For any node $v$ of $\CDAWG(T_{k,m}) = (\mathsf{V}_{T_{k,m}}, \mathsf{E}_{T_{k,m}})$, let $\D_{T_{k,m}}(v)$ denote the out-degree of $v$.
  Then, we have that $\D_{T_{k,m}}(\varepsilon) = 4$,
  $\D_{T_{k,m}}((ab)^i) = 3$ for every $1 \leq i \leq m+k$, 
  $\D_{T_{k,m}}((ab)^j) = 2$ for every $m+k+1 \leq j \leq 2m$, and $\D_{T_{k,m}}(T_{k,m})=0$.
  Therefore $\size(T_{k,m}) =5m+k+4$.

\begin{figure}[h!]
  \centering
  \includegraphics[keepaspectratio,scale=0.45]{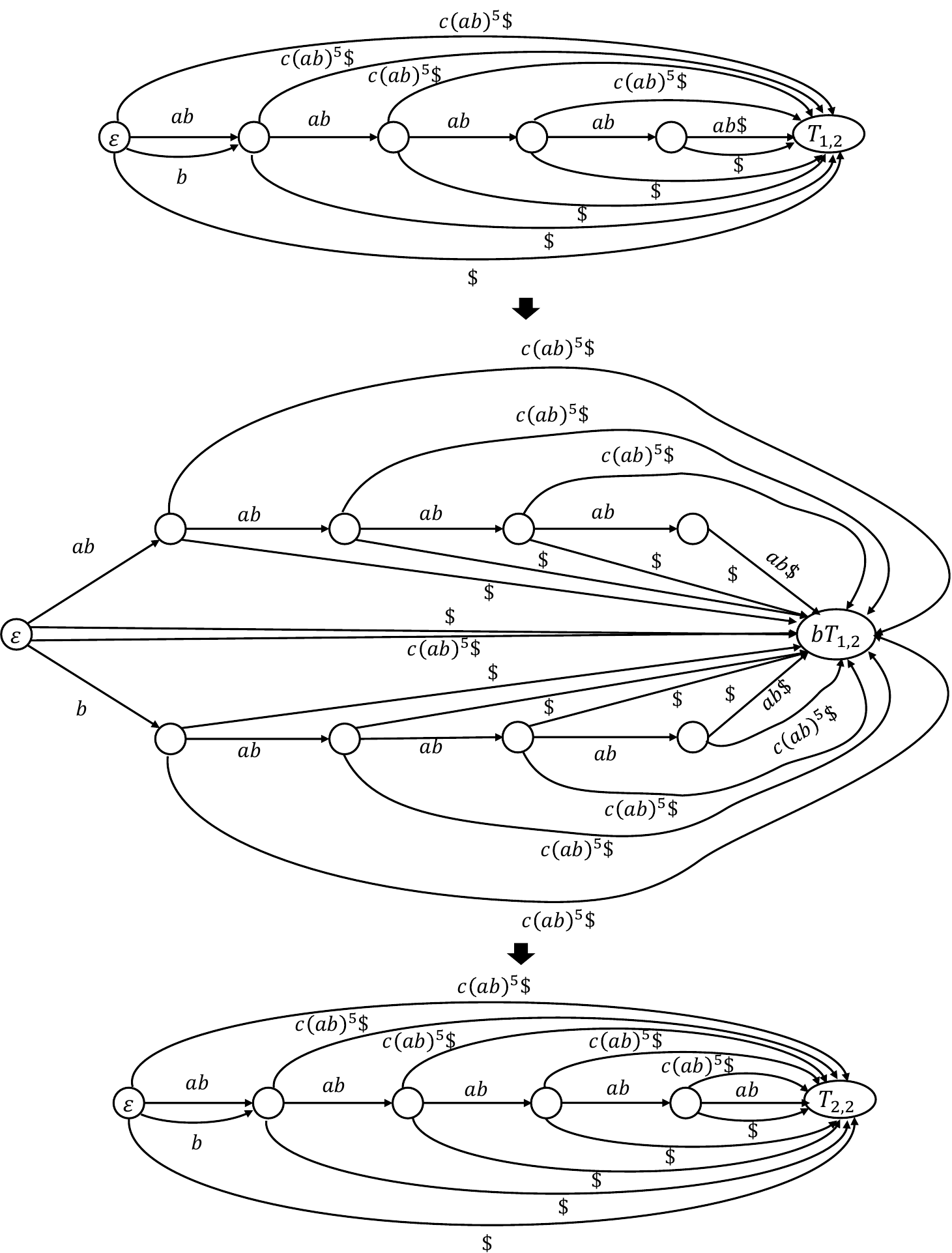}
  \caption{Illustration for the CDAWGs of strings $T_{k,m}=(ab)^3 cab(ab)^4\$$, $bT_{k,m}=b(ab)^3 cab(ab)^4\$$, 
  and $T_{k+1,m}=(ab)^4 cab (ab)^4\$$ with $k = 1, m = 2$.}
  \label{fig:leftonline_lowerbound}
\end{figure}
  
  Let us now prepend character $b$ to $T_{k,m}$ and obtain
  $$T_{k+1,m} = bT_{k,m} = b(ab)^{m+k} c(ab)^{2m}\$. $$ 
  It is clear that $bT_{k,m}=T_{m,m}[2(m-k)..n_m]$. 
  We have that
  \begin{eqnarray*}
    \M(bT_{k,m}) & = & \{\varepsilon,ab,(ab)^2,...,(ab)^{2m},b,bab,b(ab)^2,...,b(ab)^{m+k},bT_{k,m}\} \\
    & = & \left( \M(T_{k,m}) \setminus \{T_{k,m}\} \right) \cup \{b,bab,b(ab)^2,...,b(ab)^{m+k}\} \cup \{bT_{k,m}\},
  \end{eqnarray*}
  and that $\D_{bT_{k,m}}(\varepsilon)=4$,
  $\D_{bT_{k,m}}(b) = 3$,
  $\D_{bT_{k,m}}((ab)^i) = \D_{bT_{k,m}}(b(ab)^i)= 3$ for every $1 \leq i \leq m+k$,
  $\D_{bT_{k,m}}(b(ab)^j) = 2$ for every $m+k+1 \leq j\leq 2m$, 
  and $\D_{bT_{k,m}}(bT_{k,m}) = 0$. 
  Thus $\size(bT_{k,m}) = 8m+4k+7$.
  Therefore, building $\CDAWG(T_{k+1,m})$ from $\CDAWG(T_{k,m})$
  requires to \emph{add} $|\size(T_{k+1,m}) - \size(T_{k,m})| = 3m+3k+3 = \Omega(m)$ new edges (see the first step of Figure~\ref{fig:leftonline_lowerbound} for illustration).

  Let us move on to the next step, where we prepend character $a$ to $bT_{k,m}$ and obtain 
  $T_{k+1,m} = abT_{k,m} = ab(ab)^{m+k} c(ab)^{2m}\$ $. 
  Note that $abT_{k,m}=T_{k+1,m} = T_m[2(m-k)-1..n_m]$,
  and $\M(T_{k+1,m})=\{\varepsilon, ab, (ab)^2,...,(ab)^{2m},T_{k+1,m}\}$.
  We also have $\D_{T_{k+1,m}}(\varepsilon) = 4$, $\D_{T_{k+1,m}}((ab)^i) = 3$ for every $1 \leq i \leq m+k+1$, 
  $\D_{T_{k+1,m}}((ab)^j) = 2$ for every $m+k+2 \leq j \leq 2m$, and $\D_{T_{k+1,m}}(T_{k+1,m})=0$.
  This leads to $\size(T_{k+1,m}) =5m+k+5$.
  Therefore, building $\CDAWG(T_{k+1,m})$ from $\CDAWG({bT_{k,m}})$
  requires to \emph{remove} $|\size(T_{k+1,m}) - \size(bT_{k,m})| = 3m+3k+2 = \Omega(m)$ existing edges (see the second step of Figure~\ref{fig:leftonline_lowerbound} for illustration).

  This process of adding and removing $\Omega(m)$ edges in every two steps
  repeats when we update $\CDAWG(T_{k,m})$ to $\CDAWG(T_{k+1,m})$
  for every increasing $k = 1,\ldots, m-1$.
  Since $m = \Theta(n_m)$, 
  the total work for building $\CDAWG(T_m[i..n_m])$ for decreasing
  $i = n_m, \ldots, 1$ is $\Omega(m^2) = \Omega({n_m}^2)$.
\end{proof}

\begin{remark}
The linear-time algorithm of~\cite{Inenaga2005}
for \emph{rightward} online CDAWG construction maintains a slightly
modified version of the CDAWG, which becomes isomorphic to our CDAWG
when a terminal symbol $\$$ is appended to the string.
Still, our lower bound instance from Theorem~\ref{theo:lowerbound_leftonline} shows that $\$$ does not help improve the time complexity of \emph{leftward} online CDAWG construction.
\end{remark}

\subsection{Leftward online batched construction of CDAWGs}
\label{sec:batched_online}

The CDAWGs for the lower bound instance of Theorem~\ref{theo:lowerbound_leftonline}
have a periodic structure
such that
\begin{itemize}
\item $\CDAWG(T_{2j,m})$ for all even $k = 2j$ have common structures, and
\item $\CDAWG(T_{2j+1,m})$ for all odd $k = 2j+1$ have common structures.
\end{itemize}
This poses the following question: What if we allow a batched update of the CDAWG in its leftward online construction, where a string of fixed length $b > 1$ is prepended to the current string?
Namely, can we do any better when we are only to update $\CDAWG(T[1+kb..n])$ to $\CDAWG(T[1+(k-1)b..n])$ for the input string $T$ with decreasing $k = n/b, \ldots, 1$?
The next lemma however answers this question negatively:

\vspace*{1pc}
\begin{theorem} \label{theo:batched_online}
  For any fixed positive integer $b$, which devides $n$, 
  there exits a family of strings $T$ of length $n$ for which
  the total work for updating $\CDAWG(T[1+kb..n])$ to 
  $\CDAWG(T[1+(k-1)b..n])$ for decreasing $k = n/b, \ldots, 1$ $i = n, \ldots, 1$ is $\Omega(n^2)$.
\end{theorem}

\begin{proof}
  Let $t$ be any integer such that $\gcd(b,t) = 1$ and $t \geq 2$,
  and let $\Sigma = \{ \sigma_1,\sigma_2,...,\sigma_t,\#,\$\}$.
  Let
  \begin{eqnarray*}
    S & = & \sigma_1\sigma_2 \cdots \sigma_t, \\
    S'_i & = & \sigma_{t-i+1}\sigma_{t-i+2} \cdots \sigma_t,
  \end{eqnarray*}
  where $1 \leq i \leq t-1$.
  Namely, $S'_i$ is the suffix of $S$ of length $i$.

  Consider strings
  \begin{eqnarray*}
  T & = &  S^{2m} \# S^{2m+1}\$ = (\sigma_1\sigma_2 \cdots \sigma_t)^{2m} \# (\sigma_1\sigma_2 \cdots \sigma_t)^{2m+1}\$, \\
  T_k & = &  S^{m+k} \# S^{2m+1}\$ = (\sigma_1\sigma_2 \cdots \sigma_t)^{m+k} \#  (\sigma_1\sigma_2 \cdots \sigma_t)^{2m+1}\$,
  \end{eqnarray*}
  where $1 \leq k \leq m$. 
  Namely, $T_k = T[(m-k)t+1..n]$.
  Let $r = \lfloor b/t \rfloor$ and $d = b \bmod t$.
  
  In the string $T_k$, we have that
  \[
  \M(T_k)=\{\varepsilon, S, S^2, \ldots ,S^{2m},T_k\}.
  \]
  We also have $\D_{T_k}(\varepsilon) = t+2$, $\D_{T_k}(S^j) = 3$ for every $1 \leq j \leq m+k+1$, 
  $\D_{T_k}(S^j) = 2$ for every $m+k+2 \leq j \leq 2m$, and $\D_{T_k}(T_k)=0$.
  This leads to $\size(T_k) =5m+k+t+3$.

  On the other hand, in the string $S'_i T_k$ with any $1 \leq i \leq t-1$, we have that
  \[
  \M(S'_i T_k) = \{\varepsilon, S, S^2, \ldots, S^{2m}, S'_i, S'_i S, S'_i S^2, \ldots, S'_i S^{m+k}, S'_i T_k\}.
  \]
  and that $\D_{S'_i T_k}(\varepsilon)=t+2$,
  $\D_{S'_i T_k}(S'_i) = 3$,
  $\D_{S'_i T_k}(S^j) = \D_{S'_i T_k}(S'_i S^j)= 3$ for every $1 \leq j \leq m+k$,
  $\D_{S'_i T_k}(S^j) = 2$ for every $m+k+1 \leq j\leq 2m$, 
  and $\D_{S'_i T_k}(S'_i T_k) = 0$. 
  Thus $\size(S'_i T_k) = 8m+4k+t+5$.

  Now, we consider updating $\CDAWG(T[1+kb..n])$ to $\CDAWG(T[1+(k-1)b..n])$ for each fixed $k = n/b, \ldots, 1$.
  We discuss the two following particular cases from the update process:
  \begin{enumerate}
    \item Updating $\CDAWG(T_k)$ to $\CDAWG(S'_b T_{k+r})$;
    \item Updating $\CDAWG(S'_{t-d}T_k)$ to $\CDAWG(T_{k+r+1})$;
  \end{enumerate}

  In Case 1, we have $\size(T_k) = 5m+k+t+3$ and $\size(S'_b T_{k+r}) = 8m+4(k+r)+t+2$ so
  $|\size(S'_b T_{k+r}) - \size(T_k)| = 3m + k + 4r - 1 = \Omega(m)$.
  In Case 2, likewise, we have $|\size(T_{k+r+1}) - \size(S'_{t-d}T_k)| = 
  |(5m + (k+r+1) + t + 3) - (8m + 4k + t + 5)| = 3m + 3k -r+1 = \Omega(m)$.
  Therefore, $\Omega(m)$ edges are added in Case 1 and then $\Omega(m)$ edges are deleted in Case 2.

  Since $\gcd(b,t) = 1$ (which implies that $b\bmod t, 2b\bmod t, \ldots, (t-1)b\bmod t$ are all different), 
  Case 1 and Case 2 occur for every $t$ times of $k$ if $1 + kb < mt$. 
  Therefore, Case 1 and Case 2 occur at least $\lfloor m/b \rfloor = \Omega(m)$ times. 

  Since $m = \Theta(n)$, 
  the total work for updating $\CDAWG(T[1+kb..n])$ to $\CDAWG(T[1+(k-1)b..n])$ for decreasing $k = n/b, \ldots, 1$ $i = n, \ldots, 1$
   is $\Omega(m^2) = \Omega(n^2)$.
\end{proof}

\begin{example} \label{ex:batched_online}
  Set $m = 8$, $t = 5$, and $b = 4$,
  and let us consider updating string  
  $(abcde)^{8} \# (abcde)^{17}\$$
  to
  $T = (abcde)^{16} \# (abcde)^{17} \$$.
  Below we pick up the important steps during the whole process of
  the updates by prepending 4 characters at each time (see also Figure~\ref{fig:batched_online} for illustration):
  \begin{itemize}
    \item The process of updating $(abcde)^{8} \# (abcde)^{17} \$$ to $bcde(abcde)^{8} \# (abcde)^{17} \$ $ with $k = 10$ adds $\Omega(m)$ edges to the CDAWG (Case 1).
    \item The process of updating $e(abcde)^{11} \# (abcde)^{17} \$$ to $(abcde)^{12} \# (abcde)^{17}\$$ with $k = 6$ removes $\Omega(m)$ edges from the CDAWG (Case 2).
   \item The process of updating $(abcde)^{12} \# (abcde)^{17} \$$ to $bcde (abcde)^{12} \# (abcde)^{17}\$$ with $k = 5$ adds $\Omega(m)$ edges to the CDWAG (Case 1).
   \item The process of updating $e(abcde)^{15} \# (abcde)^{17} \$$ to $(abcde)^{16} \# (abcde)^{17}\$$ with $k = 1$ removes $\Omega(m)$ edges from the CDAWG (Case 2).
  \end{itemize}
In the above instance, the number of steps where $\Omega(m)$ edges are added or removed is $2$~$(= \lfloor m/b \rfloor)$ for every consecutive $5$~$(= t)$ series of $k$, namely, at $k = 10$ and $6$ for $k = 10, \ldots, 6$, and at $k = 5$ and $1$ for $k = 5, \ldots, 1$.
\end{example}

\begin{figure}[h!]
  \centering
  \includegraphics[keepaspectratio,width=1.0\linewidth]{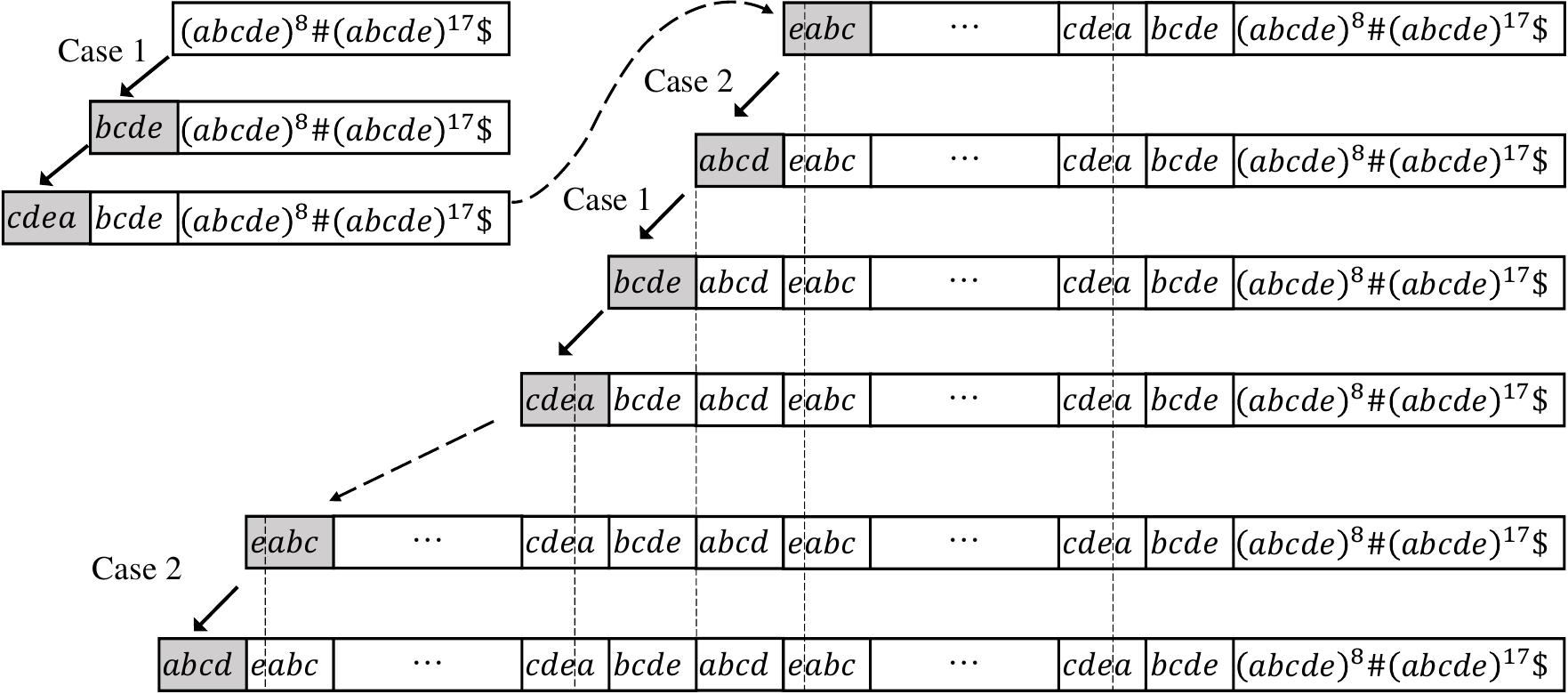}
  \caption{Illustration for Example~\ref{ex:batched_online}, where we  
    update string $(abcde)^{8}\#(abcde)^{17}\$$ to $T = (abcde)^{16} \# (abcde)^{17}\$$ by prepending a block of $4$ characters at each step.
    The characters in gray boxes are the added characters at each update.
    The vertical dashed lines exhibit the boundaries between characters $e$ and $a$.}
  \label{fig:batched_online}
\end{figure}

\section{Conclusions and further work}

This paper investigated the worst-case additive sensitivity of
the size of CDAWGs when a single-character edit operation is performed
on the left-end of the input string.
We proved that the number of new edges that appear after a left-end edit operation is at most the number of existing edges (upper bound).
We also presented (almost) matching lower bounds for all cases of left-end insertions, deletions, and substitutions.

An apparent future work is to close the small gap between our upper and lower bounds, which is at most by an additive factor of 3 (recall Table~\ref{tb:summary}).

Another intriguing open question is the sensitivity of CDAWGs when an edit operation can be performed at an arbitrary position in the string.
Our left-end sensitivity results should partly contribute to the general case, since maximal repeats that touch the edited position can be analyzed in a similar way. What remains is how to deal with maximal repeats which contain the edited position.

Belazzougui and Cunial~\cite{BelazzouguiC17} proposed the \emph{CDAWG-grammar},
which is a grammar-based string compression built on $\CDAWG(T)$.
Namely, if $\CDAWG(T)$ has $e$ edges,
then the CDAWG-grammar for $T$ is of size not greater than $e$.
In the process of building the CDAWG-grammar from the corresponding CDAWG,
every node of in-degree one is not involved in the resulting grammar. 
Thus the size of the CDAWG-grammar can be less than that of the CDAWG, and their size difference depends on each string.
Since the sensitivity of CDAWG-grammars is not well understood, it is interesting to extend our work to the sensitivity of CDAWG-grammars.

\section*{Acknowledgements}

We thank an anonymous referee who drew our attention to
the case of leftward online batched construction of CDAWGs.

This work was supported by KAKENHI grant numbers
JP21K17705~(YN), JP20H05964, 23K24808, and JP23K18466~(SI).

\bibliographystyle{abbrv}
\bibliography{ref}

\begin{thebibliography}{10}

\bibitem{AkagiFI2023}
T.~Akagi, M.~Funakoshi, and S.~Inenaga.
\newblock Sensitivity of string compressors and repetitiveness measures.
\newblock {\em Information and Computation}, 291:104999, 2023.

\bibitem{BelazzouguiC17}
D.~Belazzougui and F.~Cunial.
\newblock Fast label extraction in the {CDAWG}.
\newblock In {\em SPIRE 2017}, pages 161--175, 2017.

\bibitem{BelazzouguiCGPR15}
D.~Belazzougui, F.~Cunial, T.~Gagie, N.~Prezza, and M.~Raffinot.
\newblock Composite repetition-aware data structures.
\newblock In {\em CPM 2015}, pages 26--39, 2015.

\bibitem{Blumer1987}
A.~Blumer, J.~Blumer, D.~Haussler, R.~McConnell, and A.~Ehrenfeucht.
\newblock {Complete inverted files for efficient text retrieval and analysis}.
\newblock {\em Journal of the ACM}, 34(3):578--595, 1987.

\bibitem{BurrowsWheeler}
M.~Burrows and D.~J. Wheeler.
\newblock A block sorting lossless data compression algorithm.
\newblock Technical Report 124, Digital Equipment Corporation, 1994.

\bibitem{Crochemore1997}
M.~Crochemore and R.~V{\'{e}}rin.
\newblock {On compact directed acyclic word graphs}.
\newblock In {\em Structures in Logic and Computer Science: A Selection of
  Essays in Honor of A. Ehrenfeucht}, pages 192--211. Springer, 1997.

\bibitem{FujimaruNI23_WORDS}
H.~Fujimaru, Y.~Nakashima, and S.~Inenaga.
\newblock On sensitivity of compact directed acyclic word graphs.
\newblock In {\em {WORDS} 2023}, pages 168--180, 2023.

\bibitem{Inenaga2005}
S.~Inenaga, H.~Hoshino, A.~Shinohara, M.~Takeda, S.~Arikawa, G.~Mauri, and
  G.~Pavesi.
\newblock {On-line construction of compact directed acyclic word graphs}.
\newblock {\em Discrete Applied Mathematics}, 146(2):156--179, 2005.

\bibitem{KempaP18}
D.~Kempa and N.~Prezza.
\newblock At the roots of dictionary compression: string attractors.
\newblock In {\em STOC 2018}, pages 827--840, 2018.

\bibitem{Kociumaka2023}
T.~Kociumaka, G.~Navarro, and N.~Prezza.
\newblock Toward a definitive compressibility measure for repetitive sequences.
\newblock {\em IEEE Transactions on Information Theory}, 69(4):2074--2092,
  2023.

\bibitem{RadoszewskiR12}
J.~Radoszewski and W.~Rytter.
\newblock On the structure of compacted subword graphs of {Thue-Morse} words
  and their applications.
\newblock {\em J. Discrete Algorithms}, 11:15--24, 2012.

\bibitem{SenftD08}
M.~Senft and T.~Dvor{\'{a}}k.
\newblock Sliding {CDAWG} perfection.
\newblock In {\em {SPIRE} 2008}, pages 109--120, 2008.

\bibitem{Takagi2017}
T.~Takagi, K.~Goto, Y.~Fujishige, S.~Inenaga, and H.~Arimura.
\newblock Linear-size {CDAWG}: New repetition-aware indexing and grammar
  compression.
\newblock In {\em SPIRE 2017}, pages 304--316, 2017.

\bibitem{Takeda2000}
M.~Takeda, T.~Matsumoto, T.~Fukuda, and I.~Nanri.
\newblock Discovering characteristic expressions from literary works: a new
  text analysis method beyond n-gram statistics and {KWIC}.
\newblock In {\em Discovery Science 2000}, pages 112--126, 2000.

\bibitem{Weiner1973}
P.~Weiner.
\newblock {Linear pattern matching algorithms}.
\newblock In {\em Proceedings of the 14th Annual Symposium on Switching and
  Automata Theory}, pages 1--11, 1973.

\bibitem{LZ77}
J.~Ziv and A.~Lempel.
\newblock A universal algorithm for sequential data compression.
\newblock {\em IEEE Transactions on Information Theory}, 23(3):337--343, 1977.

\end{thebibliography}

\end{document}